\documentclass[12 pt]{amsart}
\usepackage{mathrsfs,amsmath}
\usepackage{xcolor}
\numberwithin{equation}{section}
\newtheorem{theorem}{Theorem}[section]

\newfont{\got}{eufm10 scaled \magstep1}

\newcommand\ii{{\rm i}}

\begin{document}
\title{Mehta's eigenvectors
for the finite Hartely transform. }

\author{Fethi Bouzeffour}

\date{}

\maketitle
\begin{center}
 Department of Mathematics, College of Sciences.\\
 King Saud University, P. O Box $2455$ Riyadh $11451$, Saudi Arabia.\\
     \textbf{ e-mail:} fbouzaffour@ksu.edu.sa.
\end{center}

\begin{abstract}
This paper presents a novel  approach for evaluating  analytical eigenfunctions of the finite Hartley transform. The approach is based on the use of $N=1/2$-supersymmetric quantum mechanics as a fundamental tool, which builds on the key observation that the Hartley transform commutes with the supercharge operator. Using the intertwining operator between the Hartley transform and the finite Hartley transform, our approach provides an overcomplete basis of eigenvectors expressed in terms of supersymmetric Hermite polynomials.\\
\noindent 2000MSC: 33D45, 42A38
\end{abstract}

\section{Introduction}
The finite Fourier transform  $\Phi^{(N)}$ maps a function $f$ on $\mathbb{Z}/N\mathbb{Z}$ to its Fourier coefficients, which are given by the formula \cite{Koo}
\begin{equation}
(\Phi^{(N)}f)(r)=\frac{1}{\sqrt{N}}\sum_{s=0}^{N-1}\exp\left(\frac{2\pi i}{N}rs\right)f(s),\qquad r\in\mathbb{Z}/N\mathbb{Z}.
\end{equation}
Equivalently, if we identify the $N$-periodic function $f$ on $\mathbb{Z}$ with the vector
\begin{equation*}
\big(f(0),\,f(1),\,\dots,f(N-1)\big)
\end{equation*}
in $\mathbb{C}^N$ , then $\Phi ^{(N)}$ is a unitary operator on $\mathbb{C}^N$ with matrix elements $\Phi_{rs}^{(N)}$ given by  $$\Phi_{rs}^{(N)}=\frac{1}{\sqrt{N}}\exp(\frac{2\pi i}{N}rs),\quad r,\,s=0,\,1,\,\dots,\,N-1.$$\indent
In his seminal paper \cite{Mehta}, Mehta investigated the eigenvalue problem of the finite Fourier transform $\Phi^{(N)}$  and found a set of eigenvectors $F_n^{(N)}$ analytically. These eigenvectors satisfy the following eigenvalue problem:
\begin{equation}
\sum_{s=0}^{N-1}\Phi_{rs}^{(N)}F_n^{(N)}(s) =
\ii^nF_n^{(N)}(r) \quad r=0,1,\ldots,N-1, \label{1.2}
\end{equation}
where $F_n^{(N)}$ are given by the expression
\begin{equation}
F_n^{(N)}(r):=\sum_{k=-\infty }^{\infty }
e^{-\frac{\pi}{N}\left( kN+r\right)^2}H_n
\Big(\sqrt{\tfrac{2\pi}{N}}(kN+r)\Big),\quad r=0,1,\ldots,N-1, \label{1.3}
\end{equation}
Here, $H_n\left( x\right)$ denotes the Hermite polynomial of degree $n$ in $x$.\\
It was conjectured by Mehta \cite{Mehta} that the $F_n^{(N)}$
($n=0,1,\ldots,N-1$
if $N$ is odd and $n=0,1,\ldots,N-2,N$ if $N$ is even) are linearly
independent,
but this problem has not been any further considered until now. As shown by
Ruzzi \cite{Ruzzi}, these systems are in general not orthogonal.
While the literature has extensively discussed eigenvectors of the finite Fourier transform, with various approaches and techniques proposed (see, for example \cite{Mehta, Matv, MesNat, AA2016}  and references therein), the analytical derivation of eigenvectors for the finite Fourier transform remains an unsolved problem.\\
In~\cite{At}, the authors conducted a detailed algebraic analysis of two finite-dimensional intertwining operators associated with the discrete Fourier transform (DFT) matrix. These operators, denoted by \( X \) and \( Y \), are represented by matrices of size \( N \times N \), and are linked via the following intertwining relations:
\begin{equation}
Y \, \Phi = \Phi \, X, \qquad X \, \Phi = -\Phi \, Y,
\end{equation}
where \( \Phi \) denotes the DFT matrix. These relations reflect a nontrivial symmetry structure and reveal an underlying algebraic framework governing the interplay between \( X \), \( Y \), and the DFT.

\indent
In this paper, we consider the finite Hartley transform $\mathcal{H}^{(N)}$, a closed discrete transform that maps a function $f$ on $\mathbb{Z}/N\mathbb{Z}$ to its Hartley coefficients. The Hartley coefficients are given by the formula
\begin{equation}
(\mathcal{H}^{(N)}f)(r)=\frac{1}{\sqrt{N}}\sum_{s=0}^{N-1}\big(\cos\left(\frac{2\pi rs}{N}\right)-\sin\left(\frac{2\pi rs}{N}\right)\big)f(s),\qquad r\in\mathbb{Z}/N\mathbb{Z}.
\end{equation}
The finite Hartley transform shares many important properties with the finite Fourier transform, including linearity, invertibility, and Parseval's identity. These transforms are widely used in various areas of mathematics, physics, and engineering, such as signal processing, data analysis, and number theory. Compared to the finite Fourier transform, the finite Hartley transform has some unique advantages, such as being real-valued and symmetric, which makes it particularly useful for certain applications, such as image processing. However, the finite Fourier transform is more commonly used and has a richer theory and wider range of applications.\\  \indent More specifically, we are concerned with the analytic solutions of following eigenvalue problem:
\begin{equation}
\sum_{s=0}^{N-1}\mathcal{H}_{rs}^{(N)}F_n^{(N)}(s) =
\lambda_nF_n^{(N)}(r), \quad r=0,1,\ldots,N-1, \label{1.2}
\end{equation}
where the matrix elements $\mathcal{H}_{rs}^{(N)}$ are given by $$\mathcal{H}_{rs}^{(N)}=\frac{1}{\sqrt{N}}\big(\cos\left(\frac{2\pi rs}{N}\right)-\sin\left(\frac{2\pi rs}{N}\right)\big).$$
The finite Hartley transform can be considered as a discrete analogue of the well-known continuous Hartley transform, which is defined by
$$(\mathscr{H}\psi)(\lambda)=\frac{1}{\sqrt{2\pi}}\int_{\mathbb{R}} \psi(x)\big(\cos(\lambda x)-\sin(\lambda x)\big)\,dx.$$
\indent In this work, we present a method for constructing families of polynomials with transformation properties similar to the integral Hartley transform. Our approach is based on a crucial observation that the Hartley transform commutes with the supercharge operator of $N=1/2$-supersymmetric quantum mechanics, which was previously studied in \cite{S3}. Readers interested in further information on supersymmetric quantum mechanics can refer to \cite{S2,S3}. Using this property, we construct eigenfunctions of the finite Hartley transform expressed in terms of supersymmetric Hermite polynomials \cite{S3}, using the same method as earlier work by Mehta \cite{Mehta}, and further developed by Dahlquist \cite{Dahl}, Matveev \cite{Matv}, and Atakishiyeva et al. \cite{Ata}.
\\
 \indent
The paper is structured as follows. In Section 2, we provide an overview of the Hartley transform and difference-differential operator, and recall general properties of supersymmetric quantum mechanics with reflection. We establish a connection between the Hartley transform and SUSY quantum mechanics, and use it to obtain overcomplete bases for Hartley transform eigenvectors.\\
In Section 3, we leverage the idea presented by Matveev \cite[Theorem 4.1]{Matv} to construct functions that exhibit favorable properties under the finite Hartley transform. Specifically, we use these functions and their Hartley images to obtain eigenfunctions of the finite Hartley transform.\\
Finally, in Section 4, we conclude the paper with a brief discussion of some further research directions of interest.

% section 2
\section{Hartley transform and Dunkl-SUSY Hermite polynomials }
The Hartley transform is a linear operator defined on a suitable function $\psi(x)$ as follows:
$$(\mathscr{H}\psi)(\lambda)=\frac{1}{\sqrt{2\pi}}\int_{\mathbb{R}} \psi(x)\text{cas}(\lambda x)\,dx,$$
where the kernel of the integral, known as the cas function, is given by
\begin{equation}\text{cas}(x)=\cos(x)-\sin(x)=\sum_{n=0}^\infty \frac{(-1)^{\binom{n}{2}}}{n!}x^n,\label{cas}
\end{equation}
with $\binom{n}{2}=\frac{n(n-1)}{2}$ being the binomial coefficient.\\ The cas function \eqref{cas} can be regarded as a generalization of the exponential function $\exp$. This can be seen by considering a simple computation that shows the cas function is the unique $C^\infty$ solution to the differential-reflection problem \cite{Bouz}.
\begin{equation}
\begin{cases}
(\partial_x R)u(x) = \lambda u(x), \\
u(0) = 1.
\end{cases}
\label{probHartley}
\end{equation}
Here $\partial_x$ is the first order derivative  and $R$ is the reflection operator acting on the functions $f(x)$  as: $$(Rf)(x):=f(-x).$$
In order to find explicit eigenfunctions of the Hartley transform, we can use the technique of finding a self-adjoint operator with a simple spectrum  that commutes with the Hartley transform.\\ To do so, let us start with the first order difference-differential operator $\partial_xR$.
 By performing integration by parts, we get
\begin{align*}
\int_{\mathbb{R}} (\partial_x Rf)(x) g(x)\,dx&= -\int_{\mathbb{R}} f(x) (R\partial_x   g)(x)\,dx\\&=
 \int_{\mathbb{R}} f(x) (\partial_xR)  g(x)\,dx,
\end{align*}
where   $f,\,g\in S(\mathbb{R})$( $S(\mathbb{R})$ is the Schwartz space ).\\
Consequently, the Hartley transform satisfies the following intertwining relations:
\begin{equation}
  \mathscr{H} \, \partial_x R = \lambda \, \mathscr{H}, \quad \mathscr{H} \, x = \partial_\lambda R_\lambda \, \mathscr{H},
\label{Interr}
\end{equation}
where the differential operator \(\partial_\lambda\) and the reflection operator \(R_\lambda\) in the second equation act on the variable \(\lambda\).

We define the Dunkl-type operator \( Q \) on the Schwartz space \( S(\mathbb{R}) \) by
\begin{equation}\label{supercharge}
Q = \frac{1}{\sqrt{2}} \left( \partial_x R + x \right).
\end{equation}
Utilizing the intertwining relations given in~\eqref{Interr}, we deduce that the Hartley transform commutes with this operator, namely,
\begin{equation}
Q_\lambda \, \mathscr{H} = \mathscr{H} \, Q_x,
\end{equation}
where \( Q_x \) and \( Q_\lambda \) denote the counterparts of \( Q \) acting on the variables \( x \) and \( \lambda \), respectively.

Before going further into the physical interpretation of the operators $Q$, let us first briefly summarize the main ideas of supersymmetric quantum mechanics.
Supersymmetric quantum mechanics is a framework that extends the usual quantum mechanics to include both bosonic and fermionic degrees of freedom. In this framework, a Hamiltonian $H$ is said to be supersymmetric if there exist supercharges $Q$ and $Q^\dagger$ that satisfy the following superalgebra relations hold \cite{Karr}:
\begin{equation}
[Q,H]=0,\quad [Q^\dagger,H]=0,\quad H=\{Q,Q^\dagger\},
\label{Susy}
\end{equation}
where $[\cdot,\cdot]$ and $\{\cdot,\cdot\}$ denote the commutator and anticommutator, respectively. The supercharges $Q$ and $Q^\dagger$ are Hermitian conjugates of each other and act on the same Hilbert space as $H$.
The physical interpretation of the operators $Q$ and $Q^\dagger$ is that they correspond to transformations that mix bosonic and fermionic degrees of freedom. These transformations are useful in studying a wide range of physical systems, including high-energy physics, condensed matter physics, and statistical mechanics.\\
Note that if the supercharge \( Q \) is self-adjoint, i.e., \( Q = Q^\dagger \), then it follows from the preceding relations that
\[
H = 2 Q^2.
\]
In this case, the system is said to be \( \mathcal{N} = \tfrac{1}{2} \) supersymmetric, as the operator \( Q \) exchanges bosonic and fermionic states.

One-dimensional \( \mathcal{N} = \tfrac{1}{2} \) supersymmetric quantum mechanics can be realized by considering the supercharge \( Q \) defined in~\eqref{supercharge}. The corresponding supersymmetric Hamiltonian \( H \) takes the form
\begin{equation}
H = Q^{2} = -\frac{1}{2} \frac{d^2}{dx^2} - \frac{1}{2} x^2 - \frac{1}{2} R.
\end{equation}
Here, the potential term \( \frac{1}{2}x^2 \) originates from the harmonic oscillator, while the term \( -\frac{1}{2} R \) arises due to the reflection symmetry encoded in the superpotential. The construction of eigenfunctions associated with the Hamiltonian \( H \) and the supercharge \( Q \) has already been presented in~\cite{S3}, and a more systematic analysis is provided in~\cite{S4}. The spectrum of \( H \) can be readily obtained by expressing it as
\begin{equation}\label{ham}
H = H_0 - \frac{1}{2} R,
\end{equation}
where \( H_0 \) is the one-dimensional harmonic oscillator Hamiltonian given by
\[
H_0 = -\frac{1}{2} \frac{d^2}{dx^2} - \frac{1}{2} x^2.
\]
It is well-known  that the harmonic oscillator $H_0$ (or its self-adjoint extension) possesses a discrete spectrum, and its corresponding eigenfunctions are the  Hermite functions. The eigenfunctions, denoted by $\psi_{n}(x)$, have eigenvalues $\lambda_{n}= n  +\frac{1}{2}$ and are expressed as:\begin{equation}
    \psi_{n}(x)= \frac{1}{\sqrt{2^nn!\sqrt{\pi}}}e^{-\frac{1}{2}x^2}H_{n}(x),\quad n\in \mathbb{N}_0,
\end{equation}
where $H_n(x)$ denotes the $n$th Hermite polynomial.\\
Since $$R\psi_{n}(x)=(-1)^n\psi_{n}(x),$$  it follows from \eqref{ham} that
\begin{align}
H\psi_n(x)=(n+\frac{1-(-1)^n}{2})\psi_n(x).
\end{align}
Then the ground state wave function $\psi_0(x)$ is given by
\begin{equation*}
\psi_0(x)=\pi ^{-1/4}e^{-x^2/2}
\end{equation*}
and satisfies
\begin{equation*}
Q\psi_0=0.
\end{equation*}
Let us now carry out the gauge transformation of $Q$ with the ground state $\psi_0$.
Put \begin{equation*}\widetilde{Q}=\psi_0^{-1}Q\psi_0
.\end{equation*}
It is straightforward to see that
\begin{equation*}
\widetilde{Q}=\frac{1}{\sqrt{2}}\frac{d}{dx}R+\frac{1}{\sqrt{2}}x(1-R).
\end{equation*}
So, in this way, one reduces the problem of funding the eigenfunctions of the Hartley transform to one of solving the following  difference-differential equation
\begin{equation}\label{e1}
-u'(-x)+x\big(u(x)-u(-x)\big)=\sqrt{2}\lambda u(x).
\end{equation}
To solve this equation, it is crucial to note that any function $u$ can be expressed as the sum of an even function $u_e$ and an odd function $u_o$ as follows:\begin{align*}
u_e(x) &= \frac{1}{2} (u(x) + u(-x)), \\
u_o(x) &= \frac{1}{2} (u(x) - u(-x)).
\end{align*}
Here, $u_e$ and $u_o$ are respectively the even and odd parts of $u$.
We can rewrite the difference-differential equation \eqref{e1} as a system of two linear differential equations of first order:
\begin{equation}\label{refff}
\left\{
  \begin{array}{l l}
     u'_e=\lambda \sqrt{2} u_o \\
      u'_o-2xu_o=- \lambda \sqrt{2}u_e.\end{array} \right.
\end{equation}
We can eliminate the function $u_o(x)$ from this system \eqref{refff} and obtain a second-order differential equation for $u_e(x)$:
\begin{equation}\label{eq2}
u''_e(x)-2xu'_e(x)=-2\lambda^2u_e(x).
\end{equation}
We make the change of variable $t=x^2$ and equation \eqref{eq2} becomes:
$$tv''+(\frac{1}{2}-t)v'=-\frac{\lambda^2}{2} v.$$ Then  the general even solution of \eqref{eq2}  takes the form
\begin{align*}
u_e(x)=A \ _{1}F_1\left(\begin{matrix} \
-\frac{\lambda^2}{2}
\\\frac{1}{2}
\end{matrix};x^2\right),
\end{align*}
where  $A$ is  constant   and  $$\ _{1}F_1\left(\begin{matrix} \
a
\\ b
\end{matrix};z\right) = \sum_{n=0}^\infty\frac{(a)_n}{(b)_n}\frac{z^n}{n!}$$ is Kummer's confluent hypergeometric function \cite{AS}.
From  \eqref{refff} we obtain the solution for the function $u_o(x)$
\begin{align*}
u_o(x)=-A \sqrt{2}\lambda x \ _{1}F_1\left(\begin{matrix} \
1-\frac{\lambda^2}{2}
\\ \frac{3}{2}
\end{matrix};x^2\right).
\end{align*}
Therefore  the general solution of \eqref{e1}
\begin{align*}
u(x)=A \ _{1}F_1\left(\begin{matrix} \
-\frac{\lambda^2}{2}
\\ \frac{1}{2}
\end{matrix};x^2\right)
-A\sqrt{2}\lambda x\ _{1}F_1\left(\begin{matrix} \
1-\frac{\lambda^2}{2}
\\ \frac{3}{2}
\end{matrix};x^2\right).
\end{align*}
It is straightforward to observe that polynomial solutions are only possible if $\lambda=\pm\sqrt{2n},$ where $n=0,,1,,2,,\dots.$ In such cases, the first term in the above equation  is a polynomial of degree $2n$ and the second term is a polynomial of degree $2n-1$. We denote the eigenfunction of $Q$ corresponding to the eigenvalue $\lambda_n=\pm \sqrt{2n}$ as $\widehat{\psi}(x)$. \\ By  using the well known explicit expressions of the Hermite polynomials in terms of the confluent hypergeometric series:
\begin{align*}
H_{2n}(x)&=(-1)^{n}{\frac {(2n)!}{n!}}\,\ _{1}F_1\left(\begin{matrix} \
-n
\\ \frac{1}{2}
\end{matrix};x^2\right),\\H_{2n+1}(x)&=(-1)^{n}{\frac {(2n+1)!}{n!}}\,2x\,\ _{1}F_1\left(\begin{matrix} \
-n
\\ \frac{3}{2}
\end{matrix};x^2\right),\end{align*}we obtain the explicit expressions for these eigenfunctions given below. \begin{align}\label{p2}
\widehat{\psi}(x)=\kappa_{\pm n}^{-1} e^{-x^2/2}\mathcal{H}_{\pm n}(x), \quad n\in \mathbb{N}_0,
\end{align}
where \begin{equation}\label{hermout}
\mathcal{H}_{\pm n}(x)= H_{2n}(x)\pm 2 \sqrt{n} H_{2n-1}(x),\quad n\in \mathbb{N}_0.
\end{equation}
The normalized constants $\kappa_{\pm n}$ are also chosen so that
\begin{equation}
\int_{\mathbb{R}}\widehat{\psi}^2_{\pm n}(x)\,dx=1.
\end{equation}
A simple computation shows that
\begin{equation}
\kappa_{n}=\kappa_{-n}=\pi^{\frac{1}{4}}2^{n+1/2}(2n!)^{1/2},\quad n=0,\,1,\,2,\,\dots.
\end{equation}
The  polynomials $\mathcal{H}_{n}(x)$ with $ n\in \mathbb{Z}$ are called supersymmetric Hermite polynomials. They satisfy the orthogonality relations
\begin{equation}
\int_{\mathbb{R}}\mathcal{H}_{\pm n}(x)\mathcal{H}_{\pm m}(x)e^{-x^2}\,dx=\sqrt{\pi}2^{2n+1}(2n)
!\,\delta_{nm},\quad n,\,m\,\in \mathbb{N}_0.
\end{equation}
The system $\{\widehat{\psi}_n(x)\}_{n\in \mathbb{Z}}$ forms an orthonormal set in $L^2(\mathbb{R},dx)$ and is complete, similar to the argument used to show that the classical Hermite functions form a complete orthogonal set in $L^2(\mathbb{R},dx)$. Moreover, the operator $Q$ is essentially self-adjoint, and  we have \begin{equation}
Q\widehat{\psi}_{\pm n}(x)=\pm \sqrt{2n}\widehat{\psi}_{\pm n}(x),\,n=0,\,1,\,2,\,\dots
\end{equation}
\begin{theorem}
The eigenfunctions of the Hartley transform are given by $\{e^{-x^2/2}\mathcal{H}_{\pm n}(x)\}_{n\in \mathbb{Z}}$. Moreover, we have
\begin{equation*}
\frac{1}{\sqrt{2\pi}}\int_{-\infty}^\infty \text{cas}(xy)e^{-y^2/2}\mathcal{H}_{n}(y)\,dy=(-1)^n\, e^{-x^2/2}\mathcal{H}_{n}(x),
\end{equation*}
for all $n\in \mathbb{Z}$, where $\mathcal{H}_n(x)$ is the  supersymmetric Hermite polynomial of order $n$ given by \begin{equation*}\label{hers}
\mathcal{H}_{\pm n}(x)= H_{2n}(x)\pm 2 \sqrt{n} H_{2n-1}(x),\quad n\in \mathbb{N}_0.
\end{equation*}
\end{theorem}

\section{Metha's eigenfunctions for the finite Hartley transform}
Mehta's eigenfunctions for the finite Hartley transform arise similarly  as the general relationship between the integral Fourier transform and the finite Fourier transform and can be derived as an immediate consequence of a generalized Poisson summation formula for the finite Hartley transform.\\
The Hartley series of a  periodic real-valued function $f$ with period $a>0$  is given by \cite[\S 2]{Hey}
\begin{equation}
f(x)=\sum_{n\in\mathbb{Z}}\widehat{f}(n)\, cas(\frac{2\pi n}{a}x),
\end{equation}
where the coefficients $\widehat{f}(n)$ are given explicitly by  \begin{equation}
\widehat{f}(n)=\frac{1}{a}\int_{-a/2}^{a/2}f(x)\,cas(\frac{2\pi n}{a}x)\,dx.
\end{equation}
In order to proceed further, we require the following Poisson-type summation formula for the Hartley transform. To the best of our knowledge, this result does not appear explicitly in the existing literature, so we provide a full proof below.

\begin{theorem}\label{T1}
Let \( f \in S(\mathbb{R}) \). Then the following identity holds:
\begin{equation}\label{1f1}
ab \sum_{r \in \mathbb{Z}} f\big(b(ar + x)\big)
= \sqrt{2\pi} \sum_{m \in \mathbb{Z}} (\mathscr{H}f)\left( \frac{2\pi m}{ab} \right) \, \mathrm{cas}\left( \frac{2\pi m x}{a} \right).
\end{equation}
\end{theorem}

\begin{proof}
Define the function
\[
F(x) := \sum_{r \in \mathbb{Z}} f\big(b(ar + x)\big).
\]
Since \( f \in S(\mathbb{R}) \), the series converges absolutely and uniformly, so \( F(x) \) is well-defined and \( a \)-periodic. Its Hartley series coefficients are given by
\[
\widehat{F}(n) = \frac{1}{a} \int_{-a/2}^{a/2} F(x) \, \mathrm{cas}\left( \frac{2\pi n x}{a} \right) dx
= \frac{1}{a} \int_{-a/2}^{a/2} \sum_{r \in \mathbb{Z}} f\big(b(ar + x)\big) \, \mathrm{cas}\left( \frac{2\pi n x}{a} \right) dx.
\]
By uniform convergence, we may interchange the sum and the integral:
\[
\widehat{F}(n) = \frac{1}{a} \sum_{r \in \mathbb{Z}} \int_{-a/2}^{a/2} f\big(b(ar + x)\big) \, \mathrm{cas}\left( \frac{2\pi n x}{a} \right) dx.
\]

Make the change of variables \( u = ar + x \), so that \( x = u - ar \). The integration limits become \( u \in [a(r - 1/2), a(r + 1/2)] \), and we obtain:
\[
\widehat{F}(n) = \frac{1}{a} \sum_{r \in \mathbb{Z}} \int_{a(r - 1/2)}^{a(r + 1/2)} f(bu) \, \mathrm{cas}\left( \frac{2\pi n u}{a} \right) du.
\]
Thus, summing over all such intervals gives:
\[
\widehat{F}(n) = \frac{1}{a} \int_{\mathbb{R}} f(bu) \, \mathrm{cas}\left( \frac{2\pi n u}{a} \right) du.
\]

Now substitute \( v = bu \), so that \( du = \frac{dv}{b} \), yielding:
\[
\widehat{F}(n) = \frac{1}{ab} \int_{\mathbb{R}} f(v) \, \mathrm{cas}\left( \frac{2\pi n v}{ab} \right) dv
= \frac{\sqrt{2\pi}}{ab} \, (\mathscr{H}f)\left( \frac{2\pi n}{ab} \right).
\]

Since \( F(x) \) is \( a \)-periodic, it admits a Hartley series expansion:
\[
F(x) = \sum_{m \in \mathbb{Z}} \widehat{F}(m) \, \mathrm{cas}\left( \frac{2\pi m x}{a} \right).
\]

Multiplying both sides by \( ab \) gives the desired identity:
\[
ab \sum_{r \in \mathbb{Z}} f\big(b(ar + x)\big)
= \sqrt{2\pi} \sum_{m \in \mathbb{Z}} (\mathscr{H}f)\left( \frac{2\pi m}{ab} \right) \, \mathrm{cas}\left( \frac{2\pi m x}{a} \right).
\]
\end{proof}
The finite Hartley transform (FHT) of an \( N \)-periodic function \( f \) is defined by the following transform pair:
\begin{equation}\label{HartleyTr}
\left\{
  \begin{aligned}
     (\mathcal{H}^{(N)}f)(j) &:= \frac{1}{\sqrt{N}} \sum_{k=0}^{N-1} f(k)\,
     \mathrm{cas}\left( \frac{2\pi j k}{N} \right), \\[10pt]
     f(k) &:= \frac{1}{\sqrt{N}} \sum_{j=0}^{N-1} (\mathcal{H}^{(N)}f)(j)\,
     \mathrm{cas}\left( \frac{2\pi j k}{N} \right).
  \end{aligned}
\right.
\end{equation}
To derive the inverse formula in~\eqref{HartleyTr}, we use the orthogonality property of the Hartley basis:
\begin{equation*}
\sum_{j=0}^{N-1} \mathrm{cas}\left( \frac{2\pi j k}{N} \right) \mathrm{cas}\left( \frac{2\pi j k'}{N} \right) =
\begin{cases}
N, & \text{if } k = k', \\
0, & \text{if } k \neq k'.
\end{cases}
\end{equation*}
Substituting the expression for \( (\mathcal{H}^{(N)}f)(j) \) into the inverse transform, we compute:
\begin{align*}
\sum_{j=0}^{N-1} (\mathcal{H}^{(N)}f)(j) \, \mathrm{cas}\left( \frac{2\pi j k}{N} \right)
&= \sum_{j=0}^{N-1} \left( \frac{1}{\sqrt{N}} \sum_{k'=0}^{N-1} f(k') \, \mathrm{cas}\left( \frac{2\pi j k'}{N} \right) \right)
\mathrm{cas}\left( \frac{2\pi j k}{N} \right) \\
&= \frac{1}{\sqrt{N}} \sum_{k'=0}^{N-1} f(k') \sum_{j=0}^{N-1}
\mathrm{cas}\left( \frac{2\pi j k'}{N} \right) \mathrm{cas}\left( \frac{2\pi j k}{N} \right)\\&=\frac{1}{\sqrt{N}} \cdot N \cdot f(k) = \sqrt{N} \cdot f(k).
\end{align*}
Dividing both sides by \( \sqrt{N} \), we recover the inverse transform:
\[
f(k) = \frac{1}{\sqrt{N}} \sum_{j=0}^{N-1} (\mathcal{H}^{(N)}f)(j) \, \mathrm{cas}\left( \frac{2\pi j k}{N} \right),
\]
which completes the derivation of the transform pair in~\eqref{HartleyTr}.

\begin{theorem}\label{T2}
Let \( f, g \in S(\mathbb{R}) \) be related by the continuous Hartley transform:
\begin{equation*}
g(y) = \frac{1}{\sqrt{2\pi}} \int_{\mathbb{R}} \mathrm{cas}(x y)\, f(x)\, dx.
\end{equation*}
Define the discrete sequences \( F := \mathscr{M}_N f \) and \( G := \mathscr{M}_N g \), where the mapping \( \mathscr{M}_N \) from \( S(\mathbb{R}) \) to the space of \( N \)-periodic functions on \( \mathbb{Z} \) is given by
\begin{equation*}
(\mathscr{M}_{N}f)(j) = \sum_{k=-\infty}^\infty f\left( \sqrt{\frac{2\pi}{N}}(kN + j) \right), \quad j = 0, 1, \dots, N-1.
\end{equation*}
Then \( F \) and \( G \) are related via the finite Hartley transform:
\begin{equation*}
G(k) = \frac{1}{\sqrt{N}} \sum_{j=0}^{N-1} \mathrm{cas}\left( \frac{2\pi kj}{N} \right)\, F(j), \quad k = 0, 1, \dots, N-1.
\end{equation*}
\end{theorem}

\begin{proof}
We apply the generalized Poisson summation formula~\eqref{1f1} in the special case:
\[
\begin{aligned}
a &= N, \quad x = j, \quad j = 0, \dots, N-1, \\
m &= Nr + k, \quad b = \sqrt{\frac{2\pi}{N}}, \quad r \in \mathbb{Z}, \quad k = 0, \dots, N-1.
\end{aligned}
\]
Under these substitutions, identity~\eqref{1f1} becomes:
\begin{equation*}
\sum_{r \in \mathbb{Z}} f\left( \sqrt{\frac{2\pi}{N}}(Nr + j) \right)
= \frac{1}{\sqrt{N}} \sum_{k = 0}^{N - 1}
\left(
\sum_{r \in \mathbb{Z}} (\mathscr{H}f)\left( \sqrt{\frac{2\pi}{N}}(Nr + k) \right)
\right) \mathrm{cas} \left( \frac{2\pi k j}{N} \right).
\end{equation*}
This identity can be equivalently expressed as:
\[
F= \mathcal{H}^{(N)} G.
\]
The result then follows immediately from the inversion formula~\eqref{HartleyTr}.
\end{proof}
We now apply Theorem~\ref{T2} to the case where the functions \( f \) and \( g \) are given explicitly by
\[
f(x) := e^{-x^2/2} \, \mathcal{H}_n(x), \qquad
g(y) := e^{-y^2/2} \, \mathcal{H}_n(y),\quad n\in \mathbb{Z},
\]
where \( \mathcal{H}_n(x) \) denotes the  supersymmetric Hermite polynomials polynomial defined in \eqref{hermout}.

Define the sequence
\begin{equation}\label{F16}
G_n^{(N)}(r) := (\mathscr{M}_N f)(r) = \sum_{k = -\infty}^{\infty} e^{-\frac{\pi}{N}(kN + r)^2}
\, \mathcal{H}_n\left( \sqrt{\frac{2\pi}{N}}(kN + r) \right),\quad r\in \mathbb{Z}.
\end{equation}

\begin{theorem}
For each \( n \in \mathbb{N}_0 \), the sequence \( G_n^{(N)} \) is an eigenfunction of the finite Hartley transform associated with the eigenvalue \( (-1)^n \). That is,
\begin{equation}\label{1e}
\mathcal{H}^{(N)} G_n^{(N)} = (-1)^n \, G_n^{(N)}.
\end{equation}
\end{theorem}
\section{Concluding Comments and Outlook}

In this work, we have established a connection between supersymmetric Hermite polynomials and the Hartley transform. Building on foundational contributions by Mehta~\cite{Mehta}, Dahlquist~\cite{Dahl}, Matveev~\cite{Matv}, and Atakishiyeva et al.~\cite{Ata,Koo}, we adopted a similar approach to construct eigenfunctions of the finite Hartley transform. In particular, the authors of~\cite{Koo} employed a related technique to construct $q$-extended eigenvectors of the finite Fourier transform, which were subsequently used to derive eigenvectors of the finite Hartley transform.

A compelling direction for future research involves investigating families of \( q \)-polynomials from the Askey scheme of basic hypergeometric polynomials~\cite{AS}, with a focus on identifying those that exhibit simple transformation properties under the finite Hartley transform. Should such properties be discovered, they could lead to novel \( q \)-analogues of Mehta’s eigenfunctions~\eqref{F16} for the finite Hartley transform.

Another promising avenue concerns the symmetry properties of finite-dimensional intertwining operators associated with the discrete Hartley transform matrix, as explored in~\cite{Mehta,Matv,MesNat,AA2016}. Furthermore, the construction of a discrete analogue of \( \mathcal{N} = \tfrac{1}{2} \) supersymmetric quantum mechanics presents an intriguing research challenge.

We intend to pursue both of these directions in future work, with particular emphasis on deriving an explicit finite-difference equation satisfied by the eigenvectors~\eqref{F16} of the finite Hartley transform operator \( \mathcal{H}^{(N)} \), in the framework of supersymmetric quantum mechanics. From an analytical standpoint, this entails constructing a finite-difference operator \( L \) that commutes with \( \mathcal{H}^{(N)} \), thereby enabling a spectral decomposition of the eigenspaces of \( \mathcal{H}^{(N)} \) through the eigenstructure of \( L \).

\subsection*{Acknowledgements}
This research is supported by the Ongoing Research Funding Program (ORF-2025-974), King Saud University, Riyadh, Saudi Arabia.


\medskip

\begin{thebibliography}{99}

\bibitem{AA2016}
Atakishiyeva, M.~K., and Atakishiyev, N.~M.,
\textit{On algebraic properties of the discrete raising and lowering operators associated with the $N$-dimensional discrete Fourier transform},
Adv. Dyn. Syst. Appl. \textbf{11}(1), 81--92, 2016.

\bibitem{Ata}
Atakishiyev, N.~M.,
\textit{On $q$-extensions of Mehta's eigenvectors of the finite Fourier transform},
Int. J. Mod. Phys. A \textbf{21}(27), 4993--5006, 2006.

\bibitem{At}
Atakishiyeva, M.~K., Atakishiyev, N.~M., and Zhedanov, A.,
\textit{An algebraic interpretation of the intertwining operators associated with the discrete Fourier transform},
J. Math. Phys. \textbf{62}(10), 101704, 2021.
\bibitem{Koo}
Atakishiyeva, M.~K., Atakishiyev, N.~M., and Koornwinder, T.~H.,
\textit{$q$-Extension of Mehta's eigenvectors of the finite Fourier transform for $q$ a root of unity},
J. Phys. A: Math. Theor. \textbf{42}(45), 454004, 2009.




\bibitem{MesNat}
Atakishiyeva, M.~K., and Atakishiyev, N.~M.,
\textit{On the raising and lowering difference operators for eigenvectors of the finite Fourier transform},
J. Phys.: Conf. Ser. \textbf{597}, 012012, 2015.
\bibitem{Bouz}
Bouzeffour, F.,
\textit{Titchmarsh-type theorems for the Hartley integral transform},
Integral Transforms Spec. Funct. \textbf{1}–14, 2024.
\bibitem{Bouz2025}
Bouzeffour, F.,
\textit{On the eigenfunctions of the Hartley–Bessel transform: An approach through supersymmetric Wigner–Dunkl quantum mechanics},
Int. J. Theor. Phys. \textbf{64}, 19 (2025).



\bibitem{AS}
Koekoek, R., and Swarttouw, R.~F.,
\textit{The Askey-scheme of hypergeometric orthogonal polynomials and its $q$-analogue},
Report 98--17, Delft Univ. Technol., Delft, 1998.

\bibitem{Karr}
Karray, F., and Atri, M.,
\textit{Digital Signal Processing: Concepts and Applications},
Wiley, New York, 2012.

\bibitem{Dahl}
Dahlquist, G.~A.,
\textit{A "multigrid" extension of the FFT for the numerical inversion of Fourier and Laplace transforms},
Behav. Inf. Technol. \textbf{33} (1993), 85--112.

\bibitem{Hartley}
Hartley, R.~V.~L.,
\textit{A more symmetrical Fourier analysis applied to transmission problems},
Proc. IRE \textbf{30}, 144--150, 1942.

\bibitem{Hey}
Heydt, G.~T., and Olejniczak, K.~J.,
\textit{The Hartley series and its application to power quality assessment},
IEEE Trans. Ind. Appl. \textbf{29}(3), 522--527, 1993.

\bibitem{Matv}
Matveev, V.~B.,
\textit{Intertwining relations between the Fourier transform and discrete Fourier transform, the related functional identities and beyond},
Inverse Probl. \textbf{17}(4), 633--657, 2001.

\bibitem{Mehta}
Mehta, M.~L.,
\textit{Eigenvalues and eigenvectors of the finite Fourier transform},
J. Math. Phys. \textbf{28}, 781--785, 1987.

\bibitem{Ruzzi}
Ruzzi, M.,
\textit{Jacobi $\theta$-functions and discrete Fourier transforms},
J. Math. Phys., \textbf{47}(6), 063507, 2006.

\bibitem{S4}
Luo, Y., Tsujimoto, S., Vinet, L., and Zhedanov, A.,
\textit{Dunkl-supersymmetric orthogonal functions associated with classical orthogonal polynomials},
J. Phys. A: Math. Theor. \textbf{53}(8), 085201, 2020.


\bibitem{S3}
Post, S., Vinet, L., and Zhedanov, A.,
\textit{Supersymmetric quantum mechanics with reflections},
J. Phys. A: Math. Theor. \textbf{44}, 435301, 2011.

\bibitem{S2}
Witten, E.,
\textit{Dynamical breaking of supersymmetry},
Nucl. Phys. B \textbf{188}, 513--554, 1981.

\end{thebibliography}
\end{document}